\documentclass[pra,twoside,twocolumn,longbibliography]{revtex4-2}
\usepackage{amsmath}
\usepackage{amssymb}
\usepackage{amsfonts}
\usepackage{amsthm}
\newcommand{\vecg}{\boldsymbol}
\renewcommand{\vec}{\textbf}
\newcommand{\ket}[1]{|#1\rangle}
\newcommand{\bra}[1]{\langle#1|}

\DeclareMathOperator{\tr}{Tr} 
\usepackage{subcaption}
\usepackage{graphicx}
\usepackage{color}
\newtheorem{theorem}{Theorem}
\newtheorem{lemma}{Lemma}

\begin{document}

\title{Violation of Bell inequalities in $2\times3$ dimensional systems}

\author{Pawe{\l}{} Caban}
\email{Pawel.Caban@uni.lodz.pl (corresponding author)}
\affiliation{Department of Theoretical Physics,
University of {\L}{\'o}d{\'z}\\
Pomorska 149/153, 90-236 {\L}{\'o}d{\'z}, Poland}
\author{Pawe{\l}{} Horodecki}
\email{Pawel.Horodecki@ug.edu.pl}
\affiliation{International Centre for Theory of Quantum Technologies (ICTQT),
University of Gdańsk,\\
Jana Bażynskiego 8, 80-309 Gdańsk, Poland\\
and\\
Faculty of Applied Physics and Mathematics, Gdańsk University of Technology,\\
Gabriela Narutowicza 11/12, 80-233 Gdańsk, Poland}

\date{\today}

\begin{abstract}
We consider the Clauser-Horn (CH) inequality for a
qubit-qutrit system.
We derive the necessary and sufficient conditions for the 
violation 
of the inequality as well as some sufficient conditions. 
Remarkably, we demonstrate the importance of local parameters in 
violation of the inequality. In other words, 
there are some 
families of mixed states violating the inequality for which
the correlation part alone is useless in the Bell-CH test.
\end{abstract}
\maketitle

\section{Introduction}

Quantum entanglement \cite{Schrodinger_1935_steering,cab_EPR1935}
is a surprising feature that most naturally manifests itself 
on the level of quantum pure states since in this case global
disorder exceeds the local one. In the case of mixed states, the
situation is much more complicated \cite{HHHH2009}, however
statistically, the essence of the phenomenon 
remains true to some extent.
The discovery of Bell inequalities \cite{cab_Bell1964,cab_CHSH1969}
shows that correlations represented by quantum
entanglement do not allow for description in  local and realistic
terms (see \cite{BCPSW2014} for the review).
After impressive developments 
\cite{FC1972,AGR1981,AGR1982,ADR1982,WJSWZ1998}
the violation of the inequalities have been finally confirmed
experimentally under the assumption of random inputs of measuring
devices \cite{GVetal2015,HBDetal2015,SMetal2015}.
While the violation is very easy to be checked on the paper for 
any pure state \cite{Gisin1991_entangled-violates-Bell} 
it is not such for general mixed states. 
In the case of the Clauser--Horne--Shimony--Holt (CHSH)
inequality \cite{cab_CHSH1969}
the easy condition has been provided \cite{HHH1995_Bell-violation} 
in terms of singular spectrum of the correlation tensor which 
was recently extended 
\cite{HSSS2025-Three-body-non-locality}
to the tight 3-qubit inequalities (see \cite{LPZB2004}).

It is known that the Clauser--Horne (CH) inequality is equivalent 
to CHSH in terms 
of the theoretical probabilities (although not from the perspective 
of interpretation of the experimental situation, since then, CH
requires fewer assumptions \cite{MGS2024_Stanford-Enc_Phil_Bell-theorem}). 
In this study, we focus on the mathematical part (when CHSH and CH 
are equivalent) in the case when the state is asymmetric, that is, 
it is of the qubit-qutrit type. 
We show that in this case, unlike in the two-qubit one 
the local polarization of the state may be vitally important 
for the violation.

Here, we consider the following situation:
a given two-partite quantum state 
$\rho\in \mathsf{End}({\mathcal{H}}_A\otimes{\mathcal{H}}_B)$
is distributed between Alice and Bob. They perform a certain number 
$M_A$ and $M_B$ of possible projective measurements, 
$A_i$ and $B_j$ ($i=1,\dots,M_A$, $j=1,\dots,M_B$) on their particles.
Each of these measurements has a given number of possible outcomes,
$N_A$ and $N_B$, respectively; these outcomes are denoted by $a_k$ and $b_l$.
We then calculate a set of $M_A M_B N_A N_B$ 
joint probabilities $P(a_k,b_l|A_i,B_j)$.
We are interested in determining whether such probabilities 
can be described using a local hidden variable (LHV) model. 
The set of such probabilities is convex, and its facets are 
determined by Bell inequalities.
The problem of finding relevant tight Bell inequalities for different values of 
$N_A$, $N_B$, $M_A$, $M_B$ was vastly studied in the literature,
see, e.g., \cite{BCPSW2014} and references therein.
In the case $N_A=N_B=M_A=M_B=2$ (two-qubit system and two measurements per site),
already Fine \cite{Fine1982} has shown that probabilities can be reproduced by
a LHV model if and only if they obey Clauser-Horne-Shimony-Holt
(CHSH) ineguality
for all choices of two-dimensional observables $A_1$, $A_2$, $B_1$, $B_2$ 
used by Alice and Bob.
Thus, to decide whether a given two-qubit quantum state $\rho$ can be used to
produce non-local quantum correlations one has to check whether there exist such 
observables $A_1$, $A_2$, $B_1$, $B_2$ for which the CHSH inequality is violated.
The necessary and sufficient condition for this was given in
\cite{HHH1995_Bell-violation}.

Authors of \cite{CG2004}, using linear programming methods, found 
all non-equivalent, tight 
Bell inequalities delimiting the set of lhv probabilities for small values
of $N_A$, $N_B$, $M_A$, $M_B$. 
For  $N_A=N_B=M_A=M_B=2$ they reproduced Fine's results.
In \cite{CG2004} the Bell inequality
in a Clauser-Horne (CH) form was used
\begin{multline}
	I_{CH} = P(A_1 B_1) + P(A_2 B_1) + P(A_1 B_2)\\
	-P(A_2 B_2) - P(A_1)
	- P(B_1) \le 0.
	\label{CH-inequality}
\end{multline}
In the above inequality for simplicity we denoted $P(AB)\equiv P(00|AB)$,
similarly $P(A)$ ($P(B)$) denotes the probability that when Alice (Bob) measures 
$A$ ($B$) obtains 0.
Quantum mechanical values of $I_{CH}$ are limited by
$\tfrac{1}{\sqrt{2}}-\tfrac{1}{2}$.

In the case we are interested in the present paper, i.e. for 
(qubit)$\times$(qutrit) system, two measurements per site
($N_A=2$, $N_B=3$, and $M_A=M_B=2$), it has been shown in \cite{CG2004}
that again the inequalities (\ref{CH-inequality}) delimit the set of lhv correlations.
The CH inequality (\ref{CH-inequality}) is defined for two outcome measurements. 
Thus, to use it on Bob's site (where we have three outcome measurements), 
three outcomes are mapped into two effective ones by grouping two original 
outcomes together.

\section{Characterization of local hidden variable correlations in 
	qubit-qutrit state}

The most general qubit-qutrit state can be written as
\begin{multline}
\rho= \tfrac{1}{6} (I_2 \otimes I_3 + \sum_i r_i \sigma_i\otimes I_3
+ \sum_j R_j I_2\otimes \lambda_j \\
+ \sum_{ij} T_{ij} \sigma_i\otimes\lambda_j),
\label{state-2x3-general}
\end{multline}
where $\sigma_i$, $i=1,2,3$ are Pauli matrices and 
$\lambda_j$, $j=1,\dots,8$ are Gell-Mann matrices.
The explicit form of Gell-Mann matrices and some of their properties
are given in Appendix \ref{sec:App-Gell-Mann}.
Hermicity of Pauli and Gell-Mann matrices implies that $\vec{r}\in{\mathbb{R}}^3$,
$\vec{R}\in{\mathbb{R}}^8$ and $T=[T_{ij}]$ is a real $3\times 8$ matrix.
Moreover, $\vec{r}$, $\vec{R}$ and $T$ have to fulfill additional constrains due to
semi-positive definiteness of the state $\rho$.

The probabilities appearing in the CH inequality (\ref{CH-inequality}) have the form:
\begin{align}
P(AB) & = \tr\{ \rho(P_A\otimes\Pi_B)\},
\label{probabilities-general-AB}\\
P(A) & = \tr\{ \rho(P_A\otimes I_3)\},\\
P(B) & = \tr\{ \rho(I_2\otimes\Pi_B)\},
\label{probabilities-general-B}
\end{align}
where $P_A$ and $\Pi_B$ are projectors acting in ${\mathbb{C}}^2$
and ${\mathbb{C}}^3$, respectively. 
Alice uses projectors $P_A$, $P_A^\prime$, Bob $\Pi_A$ and $\Pi_B^\prime$.
Now, we insert projectors (\ref{probabilities-general-AB}--\ref{probabilities-general-B})
into (\ref{CH-inequality})
and maximize $I_{CH}$ over all possible projectors
$P_A$, $P_A^\prime$, $\Pi_A$, $\Pi_B^\prime$.
Without loss of generality, as projectors $P_A$ and $P_A^\prime$ we can take
\begin{equation}
P_A=\tfrac{1}{2}(I_2+\vec{a}\cdot\vecg{\sigma}),\quad
P_A^\prime=\tfrac{1}{2}(I_2+\vec{a}^\prime \cdot\vecg{\sigma}),
\label{P_A}
\end{equation}
with $\vec{a},\vec{a}^\prime \in{\mathbb{R}}^3$ and maximization over all
projectors $P_A$, $P_A^\prime$ is equivalent to maximization over all unit vectors
$\vec{a},\vec{a}^\prime\in{\mathbb{R}}^3$.

The case of projectors $\Pi_B$ and $\Pi_B^\prime$ is more complicated.
These projectors act in the three dimensional space ${\mathbb{C}}^3$ thus
they can project on one or two dimensional subspace of ${\mathbb{C}}^3$.
In the former case $\tr \Pi_B=1$, in the latter one $\tr \Pi_B=2$.
As we have mentioned before, when we want to adapt three outcome
measurements to CH inequality we have to group three outcomes into
two effective ones: (1 outcome) + (2 outcomes).
$\Pi_B$ with trace 1 corresponds to the case when the result ``zero'' is connected
with one outcome while $\Pi_B$ with trace 2 to the case when the result
``zero'' is connected with a group of two outcomes.
Therefore, we are to consider the following four cases:
(i) $\tr \Pi_B=1$, $\tr \Pi_B^\prime=1$,
(ii) $\tr \Pi_B=2$, $\tr \Pi_B^\prime=2$,
(iii) $\tr \Pi_B=1$, $\tr \Pi_B^\prime=2$,
(iv) $\tr \Pi_B=2$, $\tr \Pi_B^\prime=1$.
Taking this into account and inserting (\ref{state-2x3-general},\ref{P_A})  into
(\ref{probabilities-general-AB}--\ref{probabilities-general-B}) we obtain
\begin{multline}
P(AB)  = \tfrac{1}{6} \big\{ 
(1+\vec{a}\cdot\vec{r}) \tr\Pi_B
+ \tr[(\vec{R}\cdot\vecg{\lambda})\Pi_B] \\
+ \sum_{ij} a_i T_{ij} \tr(\lambda_j \Pi_B)
\big\}
\label{probability_P_AB_intermediate}
\end{multline}
\begin{align}
P(A) & = \tfrac{1}{2}(1+\vec{a}\cdot\vec{r}),\\
P(B) & = \tfrac{1}{3} \big\{
\tr\Pi_B + \tr[(\vec{R}\cdot\vecg{\lambda})\Pi_B]
\big\}.
\end{align}
This form of probabilities leads to (cf. (\ref{CH-inequality}))
\begin{align}
I_{CH} & = 
\tfrac{1}{6} [(\vec{a}+\vec{a}^\prime)\cdot\vec{r}]
\tr\Pi_B 
+
\tfrac{1}{6} [(\vec{a}-\vec{a}^\prime)\cdot\vec{r}]
\tr\Pi_B^\prime \nonumber\\ 
& \phantom{=}
+ \tfrac{1}{6} \sum_{ij} (a_i+a_i^\prime) T_{ij} 
\tr(\lambda_j \Pi_B)\nonumber\\
& \phantom{=}
+ \tfrac{1}{6} \sum_{ij} (a_i-a_i^\prime) T_{ij} 
\tr(\lambda_j \Pi_B^\prime)\nonumber\\
& \phantom{=}
-\tfrac{1}{2} (1+\vec{a}\cdot\vec{r}).
\label{I_CH_intermediate}
\end{align}

It is very interesting, and surprising, that in this case
$I_{CH}$ depends not only on the bipartite correlation tensor
$T_{ij}$ but also on the local vector $\vec{r}$.
This point deserves some further discussion.
As it is well known (see, e.g., \cite{HHH1995_Bell-violation})
the value of the left hand side of the CHSH inequality
\begin{equation}
	I_{CHSH} = E(A_1,B_1) + E(A_2,B_1) + E(A_1, B_2) - E(A_2,B_2),
	\label{ICHSH}
\end{equation}
where the correlation function $E(A_i,B_j)$, $i,j=1,2$, in a
two-qubit state $\rho^{(2\times2)}$ is equal to 
\begin{equation}
	E(A_i,B_j) = \tr\big( \rho^{(2\times2)} A_i\otimes B_j \big),
\end{equation}
depends only on correlation tensor and parameters of observables 
used by Alice and Bob.
In a qubit case the most general two-outcome observable has
the form
\begin{equation}
	A= \vec{a}\cdot\vecg{\sigma}.
	\label{qubit-observable}
\end{equation}
Taking such a form of observables $A_i,B_j$, $i,j=1,2$
with vectors $\vec{a}$, $\vec{a}^\prime$, $\vec{b}$,
$\vec{b}^\prime$, respectively, in a general
two-qubit state
\begin{multline}
	\rho^{(2\times2)} = \tfrac{1}{4} \Big(
	I_2\otimes I_2 + (\vec{r}^A\cdot\vecg{\sigma}) \otimes I_2
	+ I_2\otimes (\vec{r}^B\cdot\vecg{\sigma})\\
	+ \sum_{ij} C_{ij} \sigma_i\otimes \sigma_j
	\Big)
	\label{state-two-qubit}
\end{multline}
one obtains \cite{HHH1995_Bell-violation}:
\begin{equation}
	I_{CHSH}^{(2\times2)} = 
	\sum_{ij} a_i C_{ij} (b_j + b_{j}^{\prime})
	+ \sum_{ij} a_i^\prime C_{ij} (b_j - b_{j}^{\prime}).
\end{equation}

Next, as it was shown in \cite{CH1974_CH-inequality,Cereceda2001},
the CH and CHSH inequalities are equivalent assuming that 
detectors are perfectly efficient and the principle of locality
holds.
Thus, it is legitimate to ask what is the form of $I_{CHSH}$ in
a qubit-qutrit system. In this case to calculate the correlation
function we need to determine the most general form of
observables. For a qubit space we take observables in the
form (\ref{qubit-observable}), i.e., 
$A_1 = \vec{a}\cdot\vecg{\sigma}$,
$A_2 = \vec{a}^\prime\cdot\vecg{\sigma}$.
Observables with outcomes $\pm1$ acting in a qutrit Hilbert space
we take in the following form:
\begin{align}
	B_1 & = (+1)\Pi_B + (-1)(I_3-\Pi_B),\\
	B_2 & = (+1)\Pi_B^\prime + (-1)(I_3-\Pi_B^\prime),
\end{align}
where $\Pi_B$, $\Pi_B^\prime$ are projectors acting in 
a 3-dimensional Hilbert space, thus all of the remarks made 
in the paragraph above Eq.~(\ref{I_CH_intermediate}) apply
to them.
With such a choice one obtains 
\begin{multline}
	I_{CHSH}^{(2\times3)} = \tfrac{2}{3} \Big\{
	[(\vec{a}+\vec{a}^\prime)\cdot\vec{r}] \tr\Pi_{B}
	+ [(\vec{a}-\vec{a}^\prime)\cdot\vec{r}] \tr\Pi_{B}^\prime\\
	+\sum_{ij} (a_i+a_i^\prime) T_{ij} \tr(\Pi_{B}\lambda_j)
	+\sum_{ij} (a_i-a_i^\prime) T_{ij} \tr(\Pi_{B}^\prime\lambda_j)\\
	-3(\vec{a}\cdot\vec{r})
	\Big\}
\end{multline}
and it is easily seen that $I_{CHSH}^{(2\times3)}\le 2$
is equivalent to $I_{CH}\le0$.
Thus we see that the dependence on the local vector $\vec{r}$
is an inherent property of $2\times3$ Bell nonlocality.

To proceed further we have to decide how to represent the most general
projectors $\Pi_B$, $\Pi_B^\prime$ in both cases, i.e. when their traces are
equal 1 or 2.
First, let us notice that when $\Pi$ is a projector 
in ${\mathbb{C}}^3$
with a definite trace then any projector
with the same trace can be obtained via the similarity transformation:
$\Pi\mapsto U\Pi U^\dagger$, $U\in SU(3)$.
Next, if $\tr\Pi=1$ then $I_3-\Pi$ is also a projector and
$\tr(I_3-\Pi)=2$.
Thus, it is enough to chose one projector with trace 1,
say $\tilde{\Pi}$, and then the most general projector $\Pi_B$
can be written as $U \tilde{\Pi} U^\dagger$ or $I_3- U \tilde{\Pi} U^\dagger$
with $U\in SU(3)$.

CH inequalities are historically connected with spin measurements, therefore 
we consider an observable corresponding to spin-1 projection on a given direction
$\vec{b}$, $\vec{b}^2=1$, i.e. the observable $\vec{b}\cdot\vec{S}$, 
where $S_i$ are standard spin-1 matrices (compare, e.g., \cite{cab_Ballentine2014}):
\begin{gather}
	\label{spin_1}
	S_1=\tfrac{1}{\sqrt{2}}
	\begin{pmatrix}
		0 & 1 & 0 \\
		1 & 0 & 1 \\
		0 & 1 & 0 \\
	\end{pmatrix},\quad
	S_2=\tfrac{i}{\sqrt{2}}
	\begin{pmatrix}
		0 & -1 & 0 \\
		1 & 0 & -1 \\
		0 & 1 & 0 \\
	\end{pmatrix},\\
	S_3=
	\begin{pmatrix}
		1 & 0 & 0 \\
		0 & 0 & 0 \\
		0 & 0 & -1 \\
	\end{pmatrix}.
\end{gather}
Spectral decomposition of $\vec{b}\cdot\vec{S}$ has the form
\begin{equation}
\vec{b}\cdot\vec{S} = (-1)\cdot \Pi_-^{\vec{b}} + 
0\cdot \Pi_0^\vec{b} + (+1)\cdot \Pi_+^{\vec{b}},
\end{equation}
with
\begin{equation}
\Pi_0^\vec{b} = I_3-(\vec{b}\cdot \vec{S})^2,\quad
\Pi_\pm^\vec{b} = \tfrac{1}{2}[(\vec{b}\cdot\vec{S})^2
\pm(\vec{b}\cdot\vec{S})].
\end{equation}
It can can be easily seen that $\tr(\vec{b}\cdot \vec{S})^2=2$.
Therefore, the most general Bob's projector with trace 2 can be written as
\begin{equation}
\Pi_B^U = U (\vec{b}\cdot \vec{S})^2 U^\dagger,
\quad U\in SU(3),
\end{equation}
while the most general projector with trace 1 as
\begin{equation}
\Pi_B^U = I_3 - U (\vec{b}\cdot \vec{S})^2 U^\dagger,
\quad U\in SU(3),
\end{equation}
where $\vec{b}$ is an arbitrary unit vector.
Inserting the above form of Bob's projectors into Eq.~(\ref{I_CH_intermediate})
we can obtain form of  $I_{CH}$ 
in the four cases (i)--(iv) enumerated above Eq.~(\ref{probability_P_AB_intermediate}).
However, taking into account that we are interested in maximizing 
$I_{CH}$ over all unit vectors 
$\vec{a},\vec{a}^\prime\in{\mathbb{R}}^3$ we can see that with simple
operations ($\vec{a}\leftrightarrow\vec{a}^\prime$, $\vec{a}\to-\vec{a}$,
$\vec{a}^\prime\to-\vec{a}^\prime$) all those four forms of $I_{CH}$
reduce to one:
\begin{multline}
I_{CH} = - \tfrac{1}{6}(3+\vec{a}\cdot\vec{r})
- \tfrac{1}{6}  \sum_{ij} (a_i + a_i^\prime) T_{ij} \tr(\lambda_j \Pi_{B}^{U})\\
- \tfrac{1}{6}  \sum_{ij} (a_i - a_i^\prime) T_{ij} \tr(\lambda_j \Pi_{B}^{U^\prime}),
\label{I_CH_intermediate-2}
\end{multline}
where $U,U^\prime\in SU(3)$.
Now we can apply the adjoint representation 
of the group $SU(3)$, $\gamma$, acting in ${\mathbb{R}}^8$
(for broader discussion see Appendix \ref{sec:App-Gell-Mann}):
\begin{equation}
	SU(3) \ni U \mapsto \gamma(U)\in SO(8), \quad
	U \lambda_i U^\dagger = \sum_j \gamma(U)_{ij} \lambda_j.
\end{equation}
Matrices $\gamma(U)$ form a proper subgroup of the whole orthogonal 
group $SO(8)$.
Thus
\begin{align}
\tr(\lambda_i \Pi_{B}^{U}) & = \tr(U^\dagger \lambda_i U (\vec{b}\cdot\vec{S})^2)
\nonumber \\
& = \sum_j \gamma_{ij}(U^\dagger) \tr(\lambda_j (\vec{b}\cdot\vec{S})^2).
\end{align}
It is convenient to introduce a vector $\tilde{\vecg{\beta}}\in{\mathbb{R}}^8$
with the following components
\begin{equation}
\tilde{\beta}_i = - \tfrac{\sqrt{3}}{2} \tr\big[\lambda_i (\vec{b}\cdot\vec{S})^2\big].
\end{equation}
Explicitly
\begin{multline}
\tilde{\vecg{\beta}} = \big(
-\tfrac{\sqrt{3}}{\sqrt{2}}b_1 b_3, -\tfrac{\sqrt{3}}{\sqrt{2}}b_2 b_3,
\tfrac{\sqrt{3}}{4}(1-3b_3^2), \tfrac{\sqrt{3}}{2}(b_2^2-b_1^2),\\
-\sqrt{3} b_1 b_2, \tfrac{\sqrt{3}}{\sqrt{2}}b_1 b_3,
\tfrac{\sqrt{3}}{\sqrt{2}}b_2 b_3, \tfrac{1}{4}(3b_3^2-1)
\big).
\end{multline}
It can be checked that
\begin{equation}
\tilde{\vecg{\beta}}^2=1,\quad
\tilde{\vecg{\beta}} \star \tilde{\vecg{\beta}} = \tilde{\vecg{\beta}},
\label{beta-properties}
\end{equation}
where the symmetric star ``$\star$'' product is defined in Eq.~(\ref{star-product-def}).
Notice that the set of vectors fulfilling conditions (\ref{beta-properties})
is a proper subset of the seven dimensional sphere $S^7\subset{\mathbb{R}}^8$.
Let us denote this set as
\begin{equation}
\mathcal{S} = \{
\vecg{\alpha} \in{\mathbb{R}}^8\colon \vecg{\alpha}^2=1,\quad
\vecg{\alpha}\star\vecg{\alpha} = \vecg{\alpha}
\}.
\label{set_S}
\end{equation}
This region $\mathcal{S}$ is invariant under the action of the matrices $\gamma(U)$ 
(c.f. Eq.~(\ref{star-product-transf})) but not under general $SO(8)$ rotations.
The action of the adjoint representation $\gamma$ is transitive on the set 
$\mathcal{S}$ \cite{AMM1997_qutrit}.

With this notation we have
\begin{equation}
\tr(\lambda_i \Pi_{B}^{U}) = - \tfrac{2}{\sqrt{3}}
\sum_j \gamma_{ij}(U^\dagger) \tilde{\beta}_j,
\end{equation}
and analogously for $\tr(\lambda_i \Pi_{B}^{U^\prime})$.
Therefore, $I_{CH}$ given in Eq.~(\ref{I_CH_intermediate-2})
can be written as
\begin{multline}
I_{CH} = -\tfrac{1}{6}(3+\vec{a}\cdot\vec{r})
+ \tfrac{1}{3\sqrt{3}} 
(\vec{a}+\vec{a}^\prime)^T T \vecg{\beta}\\
+ \tfrac{1}{3\sqrt{3}}
(\vec{a}-\vec{a}^\prime)^T T \vecg{\beta}^\prime,
\label{I_CH_intermediate-3}
\end{multline}
$U,U^\prime\in SU(3)$ and we denoted $\vecg{\beta}=\gamma(U^\dagger) \tilde{\vecg{\beta}}$, 
$\vecg{\beta}^\prime=\gamma(U^{\prime\dagger}) \tilde{\vecg{\beta}}$.

Thus, finally, we can formulate the following theorem:
\begin{theorem}
A qubit-qutrit state (\ref{state-2x3-general}) 
generates only lhv correlations
if and only if the inequality:
\begin{equation}
- \tfrac{\sqrt{3}}{2} \vec{a}\cdot\vec{r} 
+ (\vec{a}+\vec{a}^\prime)^T T \vecg{\beta}
+ (\vec{a}-\vec{a}^\prime)^T T \vecg{\beta}^\prime
\le \tfrac{3\sqrt{3}}{2}
\label{CH-inequality-r-T-beta}
\end{equation}
holds for all unit vectors $\vec{a},\vec{a}^\prime$ from ${\mathbb{R}}^3$
and all vectors $\vecg{\beta},\vecg{\beta}^\prime$
from the set  $\mathcal{S}$ defined in Eq.~(\ref{set_S}).
\label{theorem-1}
\end{theorem}

Let us denote the left-hand side of inequality (\ref{CH-inequality-r-T-beta})
by
\begin{multline}
E(\vec{r},T,\vec{a},\vec{a}^\prime,\vecg{\beta},\vecg{\beta}^\prime)\\
=
- \tfrac{\sqrt{3}}{2} \vec{a}\cdot\vec{r} 
+ (\vec{a}+\vec{a}^\prime)^T T \vecg{\beta}
+ (\vec{a}-\vec{a}^\prime)^T T \vecg{\beta}^\prime.
\end{multline}
Quantum mechanical values of $I_{CH}$ are limited by 
$\tfrac{1}{\sqrt{2}} - \tfrac{1}{2}$,
thus, quantum values of 
$E(\vec{r},T,\vec{a},\vec{a}^\prime,\vecg{\beta},\vecg{\beta}^\prime)$
are limited by $\tfrac{3\sqrt{3}}{\sqrt{2}}$.

\section{Upper bound on maximal value of
$E(\vec{r},T,\vec{a},\vec{a}^\prime,\vecg{\beta},\vecg{\beta}^\prime)$}

The following lemma, proven in Appendix \ref{sec:App-lemma}, will be useful in further 
analysis:
\begin{lemma}\label{lemma1}
Let $T$ be a real, rectangular $m\times n$ matrix. For arbitrary vectors
$\vec{x},\vec{z}\in{\mathbb{R}}^m$, $\vec{y},\vec{u}\in{\mathbb{R}}^n$
such that $\vec{y}\perp\vec{u}$
it holds
\begin{equation}
	\vec{x}^T T \vec{y} + \vec{z}^T T \vec{u}
	\le \mu_{1} ||\vec{x}||\, ||\vec{y}|| + \mu_{2} ||\vec{z}||\, ||\vec{u}|| ,
\end{equation}
where $\mu_{1},\mu_2$ are two largest singular values of the matrix $T$.
The equality holds when $\vec{x}$, $\vec{y}$ and $\vec{z}$, $\vec{u}$
are singular vectors of $T$
corresponding to the singular values $\mu_1$ and $\mu_2$, respectively.
\end{lemma}

Taking into account that $\vecg{\beta},\vecg{\beta}^\prime$ are unit vectors,
with the help of the above lemma we have
\begin{multline}
E(\vec{r},T,\vec{a},\vec{a}^\prime,\vecg{\beta},\vecg{\beta}^\prime)\\
\le -\tfrac{\sqrt{3}}{2} \vec{a}\cdot\vec{r} 
+ \mu_1 |\vec{a}+\vec{a}^\prime| +
\mu_2 |\vec{a}-\vec{a}^\prime|.
\label{E-bound-eq-1}
\end{multline}
In Appendix \ref{sec:App-maximization} we have shown that
\begin{multline}
\max_{\vec{a},\vec{a}^\prime} 
\big\{ 
-\tfrac{\sqrt{3}}{2} \vec{a}\cdot\vec{r} 
+ \mu_1 |\vec{a}+\vec{a}^\prime| +
\mu_2 |\vec{a}-\vec{a}^\prime|
\big\}\\
=
\tfrac{\sqrt{3}}{2}|\vec{r}| + 2 \sqrt{\mu_1^2 +\mu_2^2}.
\end{multline}

Therefore, we have proven the following theorem:
\begin{theorem}
Let a qubit-qutrit state be parametrized like in 
Eq.~(\ref{state-2x3-general}).
Then
\begin{multline}
\max_{\vec{a}, \vec{a}^\prime,\vecg{\beta},\vecg{\beta}^\prime}
\big\{- \tfrac{\sqrt{3}}{2} \vec{a}\cdot\vec{r} 
+ (\vec{a}+\vec{a}^\prime)^T T \vecg{\beta}
+ (\vec{a}-\vec{a}^\prime)^T T \vecg{\beta}^\prime
\big\}\\
\le
\tfrac{\sqrt{3}}{2}|\vec{r}|
+ 2 \sqrt{\mu_1^2 +\mu_2^2},
\label{I_CH-bound}
\end{multline}
where we maximize over
all unit vectors $\vec{a},\vec{a}^\prime$ from ${\mathbb{R}}^3$
and all vectors $\vecg{\beta},\vecg{\beta}^\prime$
from the set  $\mathcal{S}$ defined in Eq.~(\ref{set_S}),
$\mu_1$ and $\mu_2$ denote the two largest singular values of the matrix $T$.
\label{theorem-2}
\end{theorem}

One point should be stressed here. The Lemma \ref{lemma1} implies 
that the equality in (\ref{E-bound-eq-1}) and (\ref{I_CH-bound}) 
can be attained iff vectors $\vec{a}+\vec{a}^\prime$,
$\vec{a}-\vec{a}^\prime$ and $\vecg{\beta}$, $\vecg{\beta}^\prime$
are left and right singular vectors corresponding to the
largest singular values $\mu_1$ and $\mu_2$ of $T$.
But a vector $\vecg{\alpha}\in{\mathbb{R}}^8$ corresponds to
a projector in ${\mathbb{C}}^3$ iff $\vecg{\alpha}$ belongs to 
the set $\mathcal{S}$ defined in Eq.~(\ref{set_S}).
Therefore, for a given state $\rho$ (and consequently $T$) it
is not quaranteed that 
$E(\vec{r},T,\vec{a},\vec{a}^\prime,\vecg{\beta},\vecg{\beta}^\prime)$ can attain the value
$\tfrac{\sqrt{3}}{2}|\vec{r}| + 2 \sqrt{\mu_1^2 +\mu_2^2}$.
This contrasts with the two-qubit case.

Thus,  from theorems \ref{theorem-1} and \ref{theorem-2} we obtain the following
conclusion:
\begin{theorem}\label{theorem-3}
If for a qubit-qutrit state $\rho$ parametrized like in 
Eq.~(\ref{state-2x3-general}) holds
\begin{equation}
\tfrac{\sqrt{3}}{2}|\vec{r}|
+ 2 \sqrt{\mu_1^2 +\mu_2^2}
\le \frac{3\sqrt{3}}{2},
\label{lhv-final-condition}
\end{equation}
where $\mu_1$ and $\mu_2$ are the two largest singular values of the matrix $T$,
then the state $\rho$ is local.
\end{theorem}

\section{Examples}

\paragraph*{Example 1}
Let us consider the state discussed in \cite{BCM2025_OptimalBellQubitQutrit}:
\begin{equation}
\rho = \tfrac{1}{2} 
\begin{pmatrix}
x & 0 & 0 & 0 & x & 0\\
0 & y & 0 & 0 & 0 & y\\
0 & 0 & 1-x-y & 1-x-y & 0 & 0\\
0 & 0 & 1-x-y & 1-x-y & 0 & 0\\
x & 0 & 0 & 0 & x & 0\\
0 & y & 0 & 0 & 0 & y
\end{pmatrix},
\label{state_ex-1}
\end{equation}
where $0\le x\le 1$, $0\le y\le 1$, $x+y\le 1$. As shown in
\cite{BCM2025_OptimalBellQubitQutrit}, 
this state is entangled for all values of $x$ and $y$
except for $x=y=1/3$.

Using the parametrization given in (\ref{state-2x3-general}), the state (\ref{state_ex-1})
is represented by
\begin{align}
\vec{r} & = (0,0,0),\\
\vec{R} & = \tfrac{1}{4}\big(0, 0, 3 (1 - x - 2 y), 0, 0, 0, 0, \sqrt{3} (-1 + 3 x)\big),
\end{align}
and
\begin{widetext}
\begin{equation}
T = \tfrac{3}{2}
\begin{pmatrix}
x & 0 & 0 & 1-x-y & 0 & y & 0 & 0\\
0 & -x & 0 & 0 & 1-x-y & 0 & -y & 0 \\
0 & 0 & \tfrac{1}{2}(-1+3x) & 0 & 0 & 0 & 0 & \tfrac{\sqrt{3}}{2}(-1+x+2y)
\end{pmatrix}.
\end{equation}
\end{widetext}
The above $T$ matrix has the following non-zero singular values:
\begin{align}
\mu_1=\mu_2 & = \tfrac{3}{2}\sqrt{1 + 2 x^2 + 2 x (-1 + y) + 2 (-1 + y) y},\\
\mu_3 & = \tfrac{3}{2}\sqrt{1 + 3 x^2 + 3 x (-1 + y) + 3 (-1 + y) y},
\end{align}
and in the considered region ($0\le x\le 1$, $0\le y\le 1$, $x+y\le 1$)
\begin{equation}
\mu_1>\mu_3.
\end{equation}
In Fig.~\ref{fig1} we have depicted the region in which the state (\ref{state_ex-1})
fulfills the condition (\ref{lhv-final-condition}).

\begin{figure}
\includegraphics[width=0.95\columnwidth]{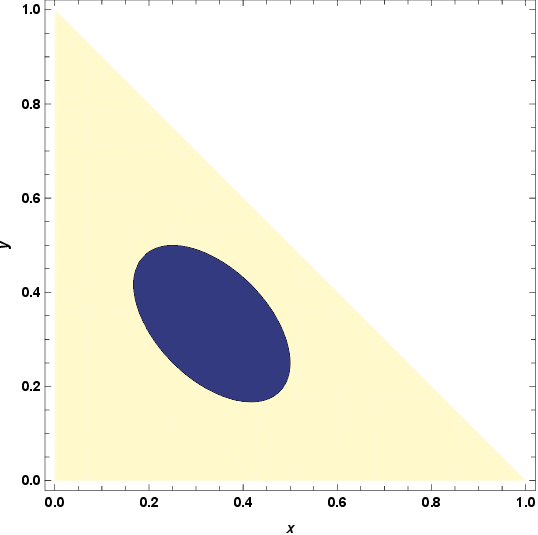}
\caption{In this figure we depicted the region of parameters (navy-blue) in which 
the condition (\ref{lhv-final-condition}) is fulfilled in the state
(\ref{state_ex-1}). }
\label{fig1}
\end{figure}

\paragraph*{Example 2}

In this example we consider a qubit-qutrit state for which 
a local vector $\vec{r}$ is nonzero.
The state has the following form
\begin{equation}
	\rho(\theta,\gamma) = \tfrac{1}{2}
	(\ket{f}\bra{f} + \ket{f(\theta)}\bra{f(\theta)}),
	\label{Ex2-state}
\end{equation}
where 
\begin{align}
	\ket{f} & = \sin\gamma \ket{0,-1} + \cos\gamma \ket{1,1}\\
	\ket{f(\theta)} & = (I_2\otimes U(\theta)) \ket{f}
\end{align}
with
\begin{equation}
	U(\theta)) = I_3 + i S_2 \sin\theta 
	- 2 S_2^2 \sin^2\tfrac{\theta}{2},
\end{equation}
and $S_2$ is given in Eq.~(\ref{spin_1}).
For the state $\rho(\theta,\gamma)$ defined in (\ref{Ex2-state}) 
we have
\begin{align}
	\vec{r} = & \big(0,0,-\cos(2\gamma)\big),\\
	\vec{R} = & \tfrac{3}{4}\big(
	-\tfrac{1}{\sqrt{2}}(\cos\theta+\cos(2\gamma))\sin\theta,0,\\
	& \tfrac{1}{4}(1+3\cos^2\theta+2\cos(2\gamma)(1+\cos\theta)),
	\tfrac{1}{2} \sin^2\theta,0,\nonumber\\
	& \tfrac{1}{\sqrt{2}}(\cos\theta-\cos(2\gamma))\sin\theta,0,
	\nonumber\\
	& \tfrac{1}{4\sqrt{3}}(-1-3\cos^2\theta
	+6\cos(2\gamma)(1+\cos\theta))
	\big).\nonumber
\end{align}
and
\begin{align}
	T_{11} & = \tfrac{3}{8\sqrt{2}}\sin(2\gamma) \sin(2\theta),\\
	T_{13} & = \tfrac{9}{16} \sin(2\gamma) \sin^2\theta,\\
	T_{14} & = \tfrac{3}{16}\big[7+\cos(2\theta)\big]\sin(2\gamma),\\
	T_{16} & = -\tfrac{3}{8\sqrt{2}}\sin(2\gamma) \sin(2\theta),\\
	T_{18} & = -\tfrac{3\sqrt{3}}{16} \sin(2\gamma) \sin^2\theta,
\end{align}
\begin{align}
	T_{22} & = \tfrac{3}{4\sqrt{2}}\sin(2\gamma) \sin\theta,\\
	T_{25} & = \tfrac{3}{2} \sin(2\gamma) \cos^2(\tfrac{\theta}{2}),\\
	T_{27} & = -\tfrac{3}{4\sqrt{2}}\sin(2\gamma) \sin\theta,
\end{align}
\begin{align}
	T_{31} & = \tfrac{3}{4\sqrt{2}}\big[1+\cos(2\gamma)\cos\theta\big] 
	\sin\theta,\\
	T_{33} & = -\tfrac{3}{16} \big[
	4\cos^2(\tfrac{\theta}{2}) + (1+3\cos^2\theta) \cos(2\gamma)
	\big],\\
	T_{34} & = -\tfrac{3}{8} \cos(2\gamma) \sin^2\theta,\\
	T_{36} & = \tfrac{3}{4\sqrt{2}} \big[1-\cos(2\gamma)\cos\theta\big] 
	\sin\theta,\\
	T_{38} & = \tfrac{\sqrt{3}}{16} \big[
	(2+\cos^2\theta)\cos(2\gamma) - 12\cos^2(\tfrac{\theta}{2})
	\big],
\end{align}
and all other matrix elements of $T$ are equal to 0.
One can check that for this state
the value of
$\tfrac{\sqrt{3}}{2}|\vec{r}|
+ 2 \sqrt{\mu_1^2 +\mu_2^2}$ is greater then or equal to
$\frac{3\sqrt{3}}{2}$
for all vales of parameters $\theta$ and $\gamma$, we have plotted 
this value in Fig.~\ref{fig2a}.

However, there are regions in the parameter space where 
$2 \sqrt{\mu_1^2 +\mu_2^2} \le \frac{3\sqrt{3}}{2}$.
We have presented the value of $2 \sqrt{\mu_1^2 +\mu_2^2}$
in~Fig.~\ref{fig2b}.
For example, for $\theta=4\pi/3$, $\gamma=\pi/12$ we have
$2 \sqrt{\mu_1^2 +\mu_2^2}=2.25$, 
$\tfrac{\sqrt{3}}{2}|\vec{r}|=0.75$
and the numerical value of $\tfrac{3\sqrt{3}}{2}$
is equal to $2.598$.

\begin{figure}
	\centering
	\begin{subfigure}{0.95\columnwidth}
		\includegraphics[width=\columnwidth]{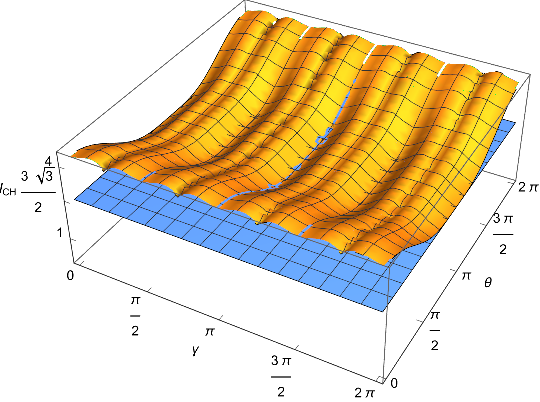}
		\caption{The value of $\tfrac{\sqrt{3}}{2}|\vec{r}|
			+ 2 \sqrt{\mu_1^2 +\mu_2^2}$
			for the state (\ref{Ex2-state}).}
		\label{fig2a}
	\end{subfigure}
	\hfill
	\begin{subfigure}{0.95\columnwidth}
		\includegraphics[width=\columnwidth]{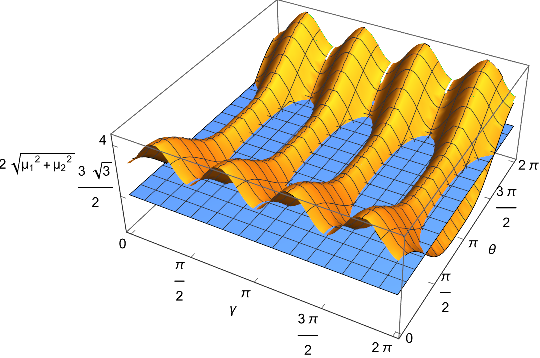}
		\caption{The value of $2 \sqrt{\mu_1^2 +\mu_2^2}$
			for the state (\ref{Ex2-state}).}
		\label{fig2b}
	\end{subfigure}
	\caption{Illustration of the role played by the local vector $\vec{r}$.}
	\label{fig2}
\end{figure}

\paragraph*{Example 3}

Authors of the paper \cite{MMH2017-2x3MES} considered families
of maximally entangled mixed states for qubit-qutrit systems.
We discuss here one of such families, namely the family of rank 2
states called true-generalized X states (TGX states):
\begin{equation}
	\rho_{\mathrm{TGX}} = \frac{1}{2}
	\begin{pmatrix}
		2 p_1 c_{\theta_1}^2 & 0 & 0 & 0 & 0 & p_1 s_{2\theta_1}\\
		0 & 2 p_2 c_{\theta_2}^2 & 0 & p_2 s_{2\theta_2} & 0 & 0\\
		0 & 0 & 0 & 0 & 0 & 0 \\
		0 & p_2 s_{2\theta_2} & 0 & 2 p_2 s_{\theta_2}^2 & 0 & 0\\
		0 & 0 & 0 & 0 & 0 & 0\\
		p_1 s_{2\theta_1} & 0 & 0 & 0 & 0 & 2 p_1 s_{\theta_1}^2
	\end{pmatrix},
	\label{state_TGX}
\end{equation}
where $p_1+p_2=1$ and $c_x=\cos x$, $s_x=\sin x$. As shown in 
\cite{MMH2017-2x3MES}, purity of this state is equal to
\begin{equation}
	P=p_1^2+p_2^2,
\end{equation}
while its negativity
\begin{multline}
	{\mathcal{N}} = - p_1 \cos^2\theta_1 
	- p_2\sin^2\theta_2 
	+ \sqrt{p_1^2 \cos^4\theta_1 + p_2^2 \sin^2(2\theta_2)}\\
	+ \sqrt{p_2^2 \sin^4\theta_2 + p_1^2 \sin^2(2\theta_1)}.
\end{multline}
For the state (\ref{state_TGX}) we have
\begin{align}
	\vec{r} = & (0,0,p_1 \cos(2\theta_1)+p_2 \cos(2\theta_2)),\\
	\vec{R} = & 
	\big(0,0,\tfrac{3}{4}[p_1+p_1\cos(2\theta_1)-2p_2\cos(2\theta_2)],
	0,0,0,0,\nonumber\\
	&\tfrac{\sqrt{3}}{4}[2-3p_1+3 p_1 \cos(2\theta_1)]\big),
\end{align}
and
\begin{align}
	T_{11} & = T_{22} = 3p_2\cos\theta_2 \sin\theta_2,\\
	T_{14} & = - T_{25} = 3p_1 \cos\theta_1 \sin\theta_1,\\
	T_{33} & = \tfrac{3}{4}(-2+3p_1+p_1\cos(2\theta_1)),\\
	T_{38} & = -\tfrac{\sqrt{3}}{4}(-3p_1 + p_1\cos(2\theta_1) 
	-2p_2\cos(2\theta_2)),
\end{align}
and all other entries of $T$ are equal to 0.
In Fig.~\ref{fig3} we considered the state $\rho_{\mathrm{TGX}}$
with minimal purity (corresponding to $p+1=1/2$). 
We have depicted the regions where the 
condition (\ref{lhv-final-condition}) is fulfilled, i.e. the state 
is surely local.

\begin{figure}
	\includegraphics[width=0.95\columnwidth]{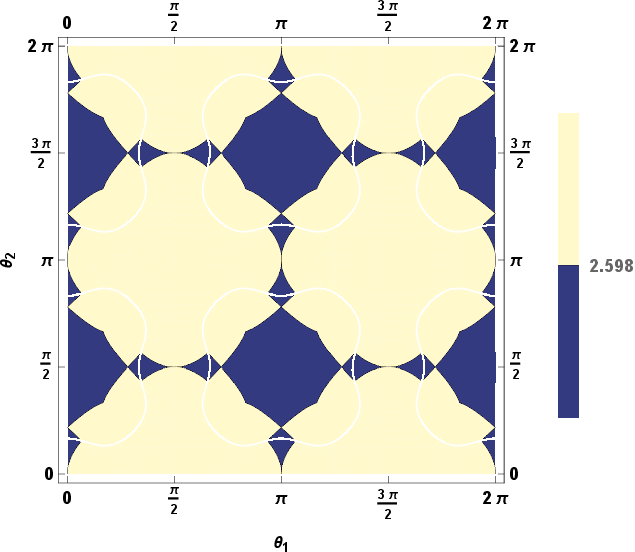}
	\caption{In this figure we depicted the region of 
	parameters (navy-blue) in which the condition 
	(\ref{lhv-final-condition}) is fulfilled in the state
	(\ref{state_TGX}) for $p_1=1/2$ (which corresponds 
	to minimal purity of the state). }
	\label{fig3}
\end{figure}

\section{Conclusions}

Lack of local realism is probably the most counterintuitive 
feature of quantum correlations. In the case of 2 qubits the 
role of local polarisations for violation of CH (CHSH ?) 
inequality  was absent. We show that in the qubit-qutrit case 
there are cases when it is extremally important for the violation. 
In fact there are mixed states for which correlation part satisfies
the Bell condition completely and this is only the local polarisation
that brakes this. Following particle physics analysis 
\cite{Barr2021,BCR2022-Bell-vector-bosons,%
BFFGM2024_QEntanglementBellViolationColliders,BCR2023_HtoZZ-anomalous,%
BCR2024_diboson-anomalous}
we shall elaborate on possible applications in the
future \cite{Caban_etal}.

\begin{acknowledgments}
We are grateful to Galileo Galilei Institute for Theoretical Physics
for hospitality during the first stages of preparation of 
this manuscript.
PH acknowledge support from the National Science Centre in Poland
under the research grant Maestro (2021/42/A/ST2/00356).
\end{acknowledgments}

\appendix
\section{Gell-Mann matrices, qutrit state and $SU(3)$ group}
\label{sec:App-Gell-Mann}

Gell-Mann matrices has the following form:
\begin{align}
& \lambda_1  = 
\begin{pmatrix}
0 & 1 & 0\\
1 & 0 & 0\\
0 & 0 & 0
\end{pmatrix},
&&\lambda_2 = 
\begin{pmatrix}
0 & -i & 0\\
i & 0 & 0\\
0 & 0 & 0
\end{pmatrix},\\
& \lambda_3  = 
\begin{pmatrix}
1 & 0 & 0\\
0 & -1 & 0\\
0 & 0 & 0
\end{pmatrix},
&&\lambda_4 = 
\begin{pmatrix}
0 & 0 & 1\\
0 & 0 & 0\\
1 & 0 & 0
\end{pmatrix},
\end{align}
\begin{align}
&\lambda_5 = 
\begin{pmatrix}
0 & 0 & -i\\
0 & 0 & 0\\
i & 0 & 0
\end{pmatrix},
&&\lambda_6 = 
\begin{pmatrix}
0 & 0 & 0\\
0 & 0 & 1\\
0 & 1 & 0
\end{pmatrix},\\
&\lambda_7 = 
\begin{pmatrix}
0 & 0 & 0\\
0 & 0 & -i\\
0 & i & 0
\end{pmatrix},
&&\lambda_8 = 
\tfrac{1}{\sqrt{3}}
\begin{pmatrix}
1 & 0 & 0\\
0 & 1 & 0\\
0 & 0 & -2
\end{pmatrix}.
\end{align}
Gell-Mann matrices are traceless, hermitian and trace-orthogonal
\begin{equation}
\tr(\lambda_i \lambda_j) = 2 \delta_{ij}.
\end{equation}
They span the Lie algebra of the $SU(3)$ group and
fulfill well-known  commutation and anticommutation relations:
\begin{align}
& [\lambda_i,\lambda_j] = 2 i \sum_k f_{ijk} \lambda_k,\\
& \{ \lambda_i,\lambda_j\} = \tfrac{4}{3} \delta_{ij} I_3 + 
2 \sum_k d_{ijk} \lambda_k,
\end{align}
where $f_{ijk}$ is completely antisymmetric while $d_{ijk}$ completely 
symmetric with the following non-vanishing independent components:
\begin{align}
& f_{123}=1, \quad f_{458}=f_{678}=\tfrac{\sqrt{3}}{2},\\
& f_{147}=f_{246}=f_{257}=f_{345}=f_{516}=f_{637}=\tfrac{1}{2},
\end{align}
and
\begin{align}
d_{118}&=d_{228}=d_{338}=-d_{888}=\tfrac{1}{\sqrt{3}},\\
d_{448}&=d_{558}=d_{668}=d_{778}=-\tfrac{1}{2\sqrt{3}}\\
d_{146}&=d_{157}=-d_{247}=d_{256}\\
&=d_{344}=d_{355}=-d_{366}=-d_{377}=\tfrac{1}{2}.
\end{align}
It also holds
\begin{align}
\lambda_i \lambda_j = \tfrac{2}{3} \delta_{ij} I_3 
+ \sum_k(d_{ijk} + i f_{ijk})\lambda_k.
\end{align}

Following e.g. \cite{AMM1997_qutrit} we consider the adjoint representation
of $SU(3)$, $\gamma$, acting in eight dimensional space ${\mathbb{R}}^8$:
\begin{equation}
SU(3) \ni U \mapsto \gamma(U), \quad
U \lambda_i U^\dagger = \sum_j \gamma(U)_{ij} \lambda_j.
\end{equation}
It holds
\begin{align}
& \gamma(U_1) \gamma(U_2) = \gamma(U_1 U_2), \quad
\gamma(U)\in SO(8),\\
& \gamma(U)_{ij} = \tfrac{1}{2} \tr(\lambda_i U \lambda_j U^\dagger). 
\end{align}
It is important to stress that matrices $\gamma(U)$ form a proper, eight
parameter subset (subgroup) of the whole twenty-eight parameter orthogonal 
group $SO(8)$.

With the help of the symbols $f_{ijk}$ and $d_{ijk}$ one can define 
antisymmetric and symmetric products of vectors in ${\mathbb{R}}^8$:
\begin{align}
&(\vecg{\alpha}\wedge\vecg{\beta})_k = \sum_{ij} f_{kij} \alpha_i \beta_j,\\
&(\vecg{\alpha}\star\vecg{\beta})_k =\sqrt{3} \sum_{ij} d_{kij} \alpha_i \beta_j.
\label{star-product-def}
\end{align}
These products have an important property that they transform like vectors in 
${\mathbb{R}}^8$ under matrices $\gamma(U)$:
\begin{align}
& (\gamma(U)\vecg{\alpha}) \wedge (\gamma(U)\vecg{\beta}) =
\gamma(U)(\vecg{\alpha}\wedge \vecg{\beta}),\\
& (\gamma(U)\vecg{\alpha}) \star (\gamma(U)\vecg{\beta}) =
\gamma(U)(\vecg{\alpha}\star \vecg{\beta}).
\label{star-product-transf}
\end{align}

Using this formalism we can describe a pure state of a qutrit.
A general qutrit state can be written as
\begin{equation}
\rho = \tfrac{1}{3} ( I_3+ \sqrt{3} \sum_i N_i \lambda_i).
\end{equation}
It can be shown (see, e.g.,\cite{AMM1997_qutrit,BN2003_Coherence-vector}) 
that conditions
$\rho^\dagger=\rho$, $\rho^2=\rho$, $\rho\ge0$, $\tr\rho=1$
are equivalent to:
\begin{equation}
\vec{N}\in{\mathbb{R}}^8,\quad
\vec{N}\cdot\vec{N}=1,\quad
\vec{N}\star\vec{N}=\vec{N}.
\label{N-qutrit-conditions}
\end{equation}
The set of vectors $\vec{N}$ fulfilling the above conditions (\ref{N-qutrit-conditions})
is a connected, simply connected four dimensional region of the sphere $S^7$.
This region is invariant under the action of the matrices $\gamma(U)$ 
(c.f. Eq.~(\ref{star-product-transf})) but not under general $SO(8)$ rotations.
Moreover, the action of the adjoint representation $\gamma$ is transitive on 
this region.

\section{Proof of the lemma}
\label{sec:App-lemma}

Let us remind first a few facts about singular values and singular vectors of 
a matrix.
Let $A$ be a real $m\times n$ matrix. Then:
\begin{list}{}{}
\item[(i)] Eigenvalues of the matrix $A^T A$ are
real and nonnegative; their square roots 
$\mu_1\ge\mu_2\ge\dots\ge\mu_n\ge0$
are singular values of $A$.
\item[(ii)] The number $r$ of non-zero singular values is equal to the rank of 
the matrix $A$, $r=\text{rank}(A)$.
\item[(iii)] Eigenvectors $\{\vec{v}_i\}$
of $A^T A$ corresponding to $\mu_{i}^{2}$ are right singular vectors of $A$:
\begin{equation}
A^T A \vec{v}_i = \mu_{i}^{2} \vec{v}_i , \quad i=1,\dots,n.
\label{App:lemma-eq1}
\end{equation}
These vectors can always be made orthonormal. Thus, from now on we will assume 
that $\{\vec{v}_1,\dots,\vec{v}_n\}$ are orthonormal; they form a basis of
${\mathbb{R}}^n$.
\item[(iv)] It holds 
\begin{equation}
A \vec{v}_i \perp A\vec{v}_j \quad \text{for}\quad i\not= j,
\label{App:lemma-eq2}
\end{equation}
and
\begin{equation}
||A \vec{v}_i || = \mu_i,\quad i=1,\dots,n.
\label{App:lemma-eq3}
\end{equation}
\item[(v)] For $i=1,\dots,r$ one can define orthonormal vectors
\begin{equation}
\vec{w}_i = \frac{1}{\mu_i} A \vec{v}_i.
\label{App:lemma-eq4}
\end{equation}
If $r<m$ then one can complete the set of vectors $\{ \vec{w}_i \}_{i=1}^{r}$
to an orthonormal basis $\{ \vec{w}_i \}_{i=1}^{m}$ of ${\mathbb{R}}^m$.
Vectors $\vec{w}_i$ are left singular vectors of $A$.
\end{list}
Using these facts we will now prove the following lemma

\begin{lemma}
Let $T$ be a real, rectangular $m\times n$ matrix. For arbitrary vectors
$\vec{x},\vec{z}\in{\mathbb{R}}^m$, $\vec{y},\vec{u}\in{\mathbb{R}}^n$
such that $\vec{y}\perp\vec{u}$
it holds
\begin{equation}
\vec{x}^T T \vec{y} + \vec{z}^T T \vec{u}
\le \mu_{1} ||\vec{x}||\, ||\vec{y}|| + \mu_{2} ||\vec{z}||\, ||\vec{u}|| ,
\end{equation}
where $\mu_{1},\mu_2$ are two largest singular values of the matrix $T$.
The equality holds when $\vec{x}$, $\vec{y}$ and $\vec{z}$, $\vec{u}$
are singular vectors of $T$
corresponding to the singular values $\mu_1$ and $\mu_2$, respectively.
\end{lemma}
\begin{proof}
We can expand vectors $\vec{y}$, $\vec{u}$ in the basis $\{\vec{v}_i\}$ and
$\vec{x}$, $\vec{z}$ in the basis $\{\vec{w}_j\}$:
\begin{equation}
\vec{x}= \sum_{i=1}^{m} x_i \vec{w}_i,\,
\vec{z}= \sum_{i=1}^{m} z_i \vec{w}_i,\,
\vec{y}= \sum_{i=1}^{n} y_i \vec{v}_i,\,
\vec{u}= \sum_{i=1}^{n} u_i \vec{v}_i.
\end{equation}
Using (\ref{App:lemma-eq1},\ref{App:lemma-eq2},\ref{App:lemma-eq4}) we obtain
\begin{equation}
\vec{x}^T T \vec{y} + \vec{z}^T T \vec{u} = 
\sum_{i=1}^r  \mu_i x_i y_i + 
\sum_{i=2}^r  \mu_i z_i u_i.
\label{App:lemma-eq8}
\end{equation}
Now, we can define the following quadratic form 
$g(\vec{x},\vec{y})=\sum_{i=1}^r \mu_i x_i y_i$ and using the Cauchy-Schwartz 
inequality for this form $g^2(\vec{x},\vec{y})\le g(\vec{x},\vec{x}) g(\vec{y},\vec{y})$
we obtain
\begin{equation}
\big[ \sum_{i=1}^r  \mu_i x_i y_i \big]^2 \le 
\big[ \sum_{i=1}^r  \mu_i x_i^2 \big]  
\big[ \sum_{i=1}^r  \mu_i y_i^2 \big].
\label{App:lemma-eq9}
\end{equation}
The maximal value of the right hand side of (\ref{App:lemma-eq9})
is attained when $x_1=||\vec{x}||$, $x_i=0$ for $i>1$ and
$y_1=||\vec{y}||$, $y_i=0$ for $i>1$ (i.e. for $\vec{x}=||\vec{x}||\vec{w}_1$
and $\vec{y}=||\vec{y}||\vec{v}_1$). Notice also that for this choice 
an inequality in (\ref{App:lemma-eq9}) becomes an equality.
Thus, we have shown that
\begin{equation}
\max_{\vec{x}, \vec{y}}\{ \vec{x}^T T \vec{y} \} = \mu_1 ||\vec{x}||\, ||\vec{y}||,
\label{App:lemma-eq10}
\end{equation}
where me maximize over vectors with constant norm, i.e. $||\vec{x}||=\text{const.}$,
$||\vec{y}||=\text{const.}$
The maximal value in (\ref{App:lemma-eq10}) is attained for 
$\vec{x}=||\vec{x}||\vec{w}_1$
and $\vec{y}=||\vec{y}||\vec{v}_1$.

Now, when we consider the whole right hand side of (\ref{App:lemma-eq8}), using 
the same method we get
\begin{multline}
\sum_{i=1}^r  \mu_i x_i y_i + 
\sum_{i=2}^r  \mu_i z_i u_i \le
\big[ \sum_{i=1}^r  \mu_i x_i^2 \big]^{1/2}  
\big[ \sum_{i=1}^r  \mu_i y_i^2 \big]^{1/2}\\
+
\big[ \sum_{i=2}^r  \mu_i z_i^2 \big]^{1/2}  
\big[ \sum_{i=2}^r  \mu_i u_i^2 \big]^{1/2},
\label{App:lemma-eq11}
\end{multline}
where the second sums on the right hand side starts with 2 due to the 
condition $\vec{y}\perp\vec{u}$
and now the maximal value of the right hand side of (\ref{App:lemma-eq11})
is attained when $x_1=||\vec{x}||$, $x_i=0$ for $i>1$, 
$y_1=||\vec{y}||$, $y_i=0$ for $i>1$,
$z_2=||\vec{z}||$, $z_i=0$ for $i\not=2$, and
$u_2=||\vec{u}||$, $u_i=0$ for $i\not=2$ 
(i.e. for $\vec{x}=||\vec{x}||\vec{w}_1$,
$\vec{y}=||\vec{y}||\vec{v}_1$,
$\vec{z}=||\vec{z}||\vec{w}_2$, and
$\vec{u}=||\vec{u}||\vec{v}_2$).
This ends the proof of the lemma.
\end{proof}

\section{Maximization procedure}
\label{sec:App-maximization}

Here we describe how to find
\begin{equation}
\max_{\vec{a},\vec{a}^\prime} 
\big\{ 
-\tfrac{\sqrt{3}}{2} \vec{a}\cdot\vec{r} 
+ \mu_1 |\vec{a}+\vec{a}^\prime| +
\mu_2 |\vec{a}-\vec{a}^\prime|
\big\},
\end{equation}
where we maximize over all unit vectors 
$\vec{a},\vec{a}^\prime\in{\mathbb{R}}^3$.

Since $(\vec{a}+\vec{a}^\prime)\cdot(\vec{a}-\vec{a}^\prime)=0$,
following \cite{HHH1995_Bell-violation}, we can always write 
\begin{equation}
	\vec{a}+\vec{a}^\prime = 2 \cos\theta\,\, \vec{c},\quad
	\vec{a}-\vec{a}^\prime = 2 \sin\theta\,\, \vec{c}^\prime,
\end{equation}
where
$\vec{c}\perp\vec{c}^\prime$,
$|\vec{c}|=|\vec{c}^\prime|=1$, $\theta\in[0,\pi/2]$,
and maximization over all unit vectors $\vec{a}$, $\vec{a}^\prime$ is equivalent 
to maximization over all mutually orthogonal, unit vectors $\vec{c}$, 
$\vec{c}^\prime$ and $\theta\in[0,\pi/2]$.

Therefore
\begin{widetext}
\begin{multline}
\max_{\vec{a},\vec{a}^\prime} 
\big\{ 
-\tfrac{\sqrt{3}}{2} \vec{a}\cdot\vec{r} 
+ \mu_1 |\vec{a}+\vec{a}^\prime| +
\mu_2 |\vec{a}-\vec{a}^\prime|
\big\}
=\max_{\vec{c},\vec{c}^\prime,\theta} 
\big\{
-\tfrac{\sqrt{3}}{2} (\vec{c}\cdot\vec{r} \cos\theta +
\vec{c}^\prime\cdot \vec{r}\sin\theta)
+ 2 \mu_1 \cos\theta + 2 \mu_2 \sin\theta)
\big\}\\
=\max_{\vec{c},\vec{c}^\prime,\theta} 
\big\{
[-\tfrac{\sqrt{3}}{2} \vec{c}\cdot\vec{r} 
+ 2 \mu_1] \cos\theta +
[-\tfrac{\sqrt{3}}{2} \vec{c}^\prime\cdot \vec{r}
+ 2 \mu_2] \sin\theta
\big\}\\
=\max_{\vec{c},\vec{c}^\prime} 
\Big\{
\sqrt{
[-\tfrac{\sqrt{3}}{2} \vec{c}\cdot\vec{r} + 2 \mu_1]^2 +
[-\tfrac{\sqrt{3}}{2} \vec{c}^\prime\cdot \vec{r} + 2 \mu_2]^2
}
\Big\}.
\label{maximization-1}
\end{multline}
Next, we can choose the third vector
$\vec{c}^{\prime\prime}$ orthogonal to both $\vec{c}$ and $\vec{c}^\prime$.
With this choice 
$x\equiv(\vec{c}\cdot\vec{r})/|\vec{r}|$,
$y\equiv(\vec{c}^{\prime}\cdot\vec{r})/|\vec{r}|$,
$z\equiv(\vec{c}^{\prime\prime}\cdot\vec{r})/|\vec{r}|$
are directional cosines of the vector $\vec{r}$ with respect to the
orthonormal basis $\{\vec{c},\vec{c}^\prime,\vec{c}^{\prime\prime}\}$
and it holds $x^2+y^2+z^2=1$.
Therefore
\begin{multline}
\max_{\vec{c},\vec{c}^\prime} 
\Big\{
\sqrt{
[-\tfrac{\sqrt{3}}{2} \vec{c}\cdot\vec{r} + 2 \mu_1]^2
+
[-\tfrac{\sqrt{3}}{2} \vec{c}^\prime\cdot \vec{r} + 2 \mu_2]^2
}
\Big\}
=\max_{x, y}
\Big\{
\sqrt{
[-\tfrac{\sqrt{3}}{2}|\vec{r}| x + 2 \mu_1]^2
+
[-\tfrac{\sqrt{3}}{2} |\vec{r}| y + 2 \mu_2]^2
}
\Big\}\\
=\max_{x, y,x\ge0,y\ge0}
\Big\{
\sqrt{
[\tfrac{\sqrt{3}}{2}|\vec{r}| x + 2 \mu_1]^2
+
[\tfrac{\sqrt{3}}{2} |\vec{r}| y + 2 \mu_2]^2
}
\Big\}\\
=\max_{x, y,x\ge0,y\ge0}
\Big\{
\sqrt{
\tfrac{3}{4}|\vec{r}|^2(x^2+y^2)
+ 2 \sqrt{3} |\vec{r}| (\mu_1 x + \mu_2 y)
+4(\mu_1^2+\mu_2^2)
}
\Big\},
\end{multline}
and $x$, $y$ are bounded by the constraint $x^2+y^2 \le 1$.
Next, simple calculation of global extremum of the above function shows that 
the maximal value is attained on the boundary
($x^2+y^2=1$) for $x=\mu_1/\sqrt{\mu_1^2 +\mu_2^2}$ and
$y=\mu_2/\sqrt{\mu_1^2 +\mu_2^2}$.
Therefore, we finally have
\begin{multline}
\max_{x, y,x\ge0,y\ge0}
\Big\{
\sqrt{
\tfrac{3}{4}|\vec{r}|^2(x^2+y^2)
+ 2 \sqrt{3} |\vec{r}| (\mu_1 x+ \mu_2 y)
+4(\mu_1^2+\mu_2^2)
}
\Big\}
=\sqrt{
\tfrac{3}{4}|\vec{r}|^2
+ 2 \sqrt{3} |\vec{r}| \sqrt{\mu_1^2+\mu_2^2}
+4 (\mu_1^2+\mu_2^2)
}\\
= \tfrac{\sqrt{3}}{2}|\vec{r}| + 2 \sqrt{\mu_1^2+\mu_2^2}.
\end{multline}
\end{widetext}


%

\end{document}